\documentclass[a4paper]{article}
\usepackage[T1]{fontenc}
\usepackage[utf8]{inputenc}
\DeclareUnicodeCharacter{0301}{\'{e}}
\usepackage{graphicx}
\usepackage{hyperref}
\hypersetup{
  setpagesize  = false,
  colorlinks   = true,    
  urlcolor     = blue,    
  linkcolor    = black,   
  citecolor    = black    
}
\usepackage{authblk}
\usepackage{geometry}
\usepackage{setspace}
\usepackage{verbatim}
\usepackage[utf8]{inputenc}
\usepackage{amsmath}
\usepackage{amssymb}
\usepackage{amsthm}
\usepackage{gensymb}
\usepackage{verbatim} 
\usepackage{ragged2e} 
\usepackage{amsfonts}
\usepackage{cancel}
\usepackage{tikz}
\usetikzlibrary{shapes, positioning, backgrounds, fit, arrows, arrows.meta}
\usepackage{xcolor}
\usepackage{xspace}
\usetikzlibrary{patterns}

\newcommand{\abstractinenglishname}{Abstract}
\newcommand{\keywordsportugues}{Palavras-chave}
\newcommand{\keywordsenglishname}{Keywords}

\renewenvironment{abstract}{%
        \begin{center}
	\begin{minipage}{14cm}
	{\textbf{\abstractname:}}
}{
        \end{minipage}
	\end{center}
}

\newenvironment{keywordsenglish}{
        \def\abstractname{\emph{\keywordsenglishname}}
	\begin{abstract}
}{
        \end{abstract}
}

\newcommand{\cclass}[1]{\ensuremath{\mathord{\rm #1}}} 

\newcommand{\rsCol}{\ruleset{Col}}

\newcommand{\ruleset}[1]{\textsc{#1}}
\newcommand{\outcomeClass}[1]{\ensuremath{\mathcal{#1}}}
\newcommand{\outP}{\outcomeClass{P}}
\newcommand{\outN}{\outcomeClass{N}}
\newcommand{\outL}{\outcomeClass{L}}
\newcommand{\outR}{\outcomeClass{R}}
\newcommand{\zeroGame}{\text{\Large 0}}
\newcommand{\impartialGameSet}[1]{\ensuremath{{}^{*\!}\{\ #1\ \} }}
\newcommand{\gameSet}[2]{\ensuremath{\left\{\ #1\ \middle|\ #2\ \right\}}}

\newcommand{\rsBCL}{\ruleset{B2CL}}

\newtheorem{theorem}{Theorem}

\title {Col is PSPACE-complete on Triangular Grids}
\author[1]{Kyle Burke} 
\affil[1]{Florida Southern College, Lakeland, FL}
\author[2]{Craig Tennenhouse}
\affil[2]{University of New England, Biddeford, ME}

\begin{document}

\maketitle
\vspace{6pt}

\begin{abstract}
We demonstrate that \rsCol{} is \cclass{PSPACE}-complete on triangular grid graphs via a reduction from \ruleset{Bounded Two-Player Constraint Logic}.  This is the most structured graph family that \rsCol{} is known to be computationally hard for.   
\end{abstract}

\begin{keywordsenglish}
combinatorial game, computational complexity, grid graphs, Col
\end{keywordsenglish}

\vspace{6pt}

\section{Introduction}

\rsCol{} is a combinatorial game played on graphs where vertices are either uncolored or colored blue or red, and no pair of adjacent vertices are both blue or both red. A turn for the two players, named Blue and Red, consists of painting an uncolored vertex their color, such that it meets the proper-coloring restriction mentioned above \cite{WinningWays:2001}\footnote{A playable version of the game is available at \url{https://kyleburke.info/DB/combGames/colTriangularGrid.html}.}.

The computational complexity of \rsCol{} has an interesting history.  It was erroneously listed as \cclass{PSPACE}-complete in the very popular Winning Ways text \cite{WinningWays:2001}.   (See page 224 of the 2001 printing.) 
This confusion continued until at least 2009, when the not-yet-proven hardness was referenced in \cite{cincotti2009three}.

\rsCol{} was first actually shown to be computationally (\cclass{PSPACE}) hard in 2015 in \cite{conf/isaac/FennerGMST15}---impressively on uncolored graphs---finally resolving this outstanding problem and retroactively confirming the previously cited result.  In 2018, \rsCol{} was shown to be \cclass{PSPACE}-complete on planar graphs (though with colored vertices) \cite{BurkeHearn2018}.

When we restrict the structure on instances, the computational complexity of a problem can never increase.  Demonstrating hardness on varying subfamilies is worthwhile and comes with new challenges.  In this paper, we determine the computational hardness of \rsCol{} on a hexagonal board, a common topology for board games.

\section{Col on Triangular Grids}

\tikzstyle{blueStone}=[draw, circle,inner sep=2, minimum size =20 pt, line width = 1pt, draw=black, fill=blue]
\tikzstyle{blueCan} = [draw, circle, inner sep = 2, minimum size = 20pt, line width = 1pt, pattern=north west lines, pattern color=blue]
\tikzstyle{redStone}=[draw, circle,inner sep=2, minimum size =20 pt, line width = 1pt, draw=black, fill=red]
\tikzstyle{redCan} = [draw, circle, inner sep = 2, minimum size = 20pt, line width = 1pt, pattern=north east lines, pattern color=red]
\tikzstyle{empty}=[draw, circle,inner sep=2, minimum size =20 pt, line width = 1pt, draw=black, fill=white]
\tikzstyle{nix}=[draw, circle, inner sep=2, minimum size =20 pt, line width = 1pt, draw=black, fill=white, path picture={ 
  \draw[black]
(path picture bounding box.south east) -- (path picture bounding box.north west) (path picture bounding box.south west) -- (path picture bounding box.north east);
}]

\subsection{Reduction Overview}

We are reducing from \ruleset{Bounded Two-Player Constraint Logic} (\rsBCL), which is \cclass{PSPACE}-complete by way of \ruleset{Positive CNF} \cite{DBLP:books/daglib/0023750}.  In order to complete the reduction, it is enough to implement gadgets for only five types of \rsBCL\ vertices: AND, OR, CHOICE, SPLIT, and VARIABLE \cite{DBLP:books/daglib/0023750}, as well as a GOAL.  Additionally, because of the strict topology of a triangular grid, we will need WIRE gadgets to connect our pieces.

Although a subset of those gadgets is adequate to reduce from \ruleset{Positive CNF}, \rsBCL\ automatically grants us the planarity necessary for our grid, as the above listed vertices are enough to resolve crossovers \cite{DBLP:books/daglib/0023750}.

The \rsCol\ position will have an even number of variable gadgets that should be played first by both players.  Then Blue will fill out the rest of the circuitry while Red plays vertices reserved for them.  If Blue is able to play sufficiently on all gadgets and play (twice) on the GOAL gadget, then they will have one remaining move after Red takes their last possible turn.

We continue by describing each of the gadgets.

\subsection{Main Gadgets}

Each of our gadgets will exist in a single 8 $\times$ 8 ``tile'' of vertices.  To describe all of our gadgets, we will show figures with the vertices drawn in one of the six states shown in the legend in Figure \ref{fig:legend}.  The first three symbols, $(a)$, $(b)$, and $(c)$, are used to show the actual color of that vertex.  The remaining three, $(d)$, $(e)$, and $(f)$, are possible configurations of empty vertices to help make things more clear for the reader, depending on what is nearby.  We don't bother with a separate notation for a vertex where either player can play.  Thus, a vertex drawn as empty, $(a)$, may actually be in one of the restricted states, $d$, $e$, or $f$, depending on the color of neighboring vertices. 

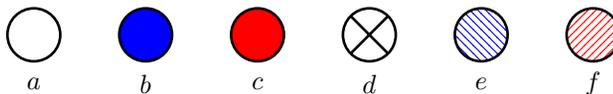
\begin{figure}[h!]
    \begin{center}
    \begin{tikzpicture}[node distance = .75cm, minimum size = .5cm, inner sep = .07cm, ultra thick]
        \node[empty] (a) [label = below:{$a$}] {};
        \node[blueStone] (b) [right=of a, label = below:{$b$}] {};
        \node[redStone] (c) [right=of b, label = below:{$c$}] {};
        \node[nix] (d) [right=of c, label = below:{$d$}] {};
        \node[blueCan] (e) [right=of d, label = below:{$e$}] {};
        \node[redCan] (f) [right=of e, label = below:{$f$}] {};
    \end{tikzpicture}
    \end{center}
    \caption{Legend.  $a$: empty vertex, $b$: vertex colored blue, $c$: vertex colored red, $d$: empty vertex adjacent to both blue and red; no one can play, $e$: empty vertex adjacent to red; Blue can play, $f$: empty vertex adjacent to blue; Red can play.}
    \label{fig:legend}
\end{figure}

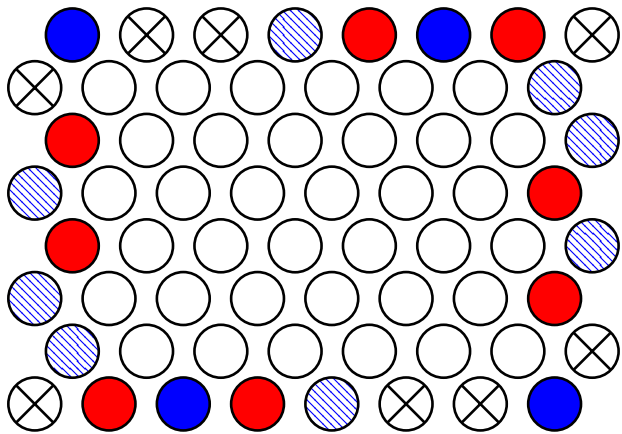
\begin{figure}[h!]
    \begin{center}
    \begin{tikzpicture}[node distance = .25cm, minimum size = .5cm, inner sep = .07cm, ultra thick]
        \node[blueStone] (t00) {};
        \node[nix] (t10) [right=of t00] {};
        \node[nix] (t20) [right=of t10] {};
        \node[blueCan] (t30) [right=of t20] {};
        \node[redStone] (t40) [right=of t30] {};
        \node[blueStone] (t50) [right=of t40] {};
        \node[redStone] (t60) [right=of t50] {};
        \node[nix] (t70) [right=of t60] {};

        \node[nix] (t01) at (-.5, -.7) {};
        \node[empty] (t11) [right=of t01] {};
        \node[empty] (t21) [right=of t11] {};
        \node[empty] (t31) [right=of t21] {};
        \node[empty] (t41) [right=of t31] {};
        \node[empty] (t51) [right=of t41] {};
        \node[empty] (t61) [right=of t51] {};
        \node[blueCan] (t71) [right=of t61] {};
        
        \node[redStone] (t02) at (0, -1.4) {};
        \node[empty] (t12) [right=of t02] {};
        \node[empty] (t22) [right=of t12] {};
        \node[empty] (t32) [right=of t22] {};
        \node[empty] (t42) [right=of t32] {};
        \node[empty] (t52) [right=of t42] {};
        \node[empty] (t62) [right=of t52] {};
        \node[blueCan] (t72) [right=of t62] {};

        \node[blueCan] (t03) at (-.5, -2.1) {};
        \node[empty] (t13) [right=of t03] {};
        \node[empty] (t23) [right=of t13] {};
        \node[empty] (t33) [right=of t23] {};
        \node[empty] (t43) [right=of t33] {};
        \node[empty] (t53) [right=of t43] {};
        \node[empty] (t63) [right=of t53] {};
        \node[redStone] (t73) [right=of t63] {};
        
        \node[redStone] (t04) at (0, -2.8) {};
        \node[empty] (t14) [right=of t04] {};
        \node[empty] (t24) [right=of t14] {};
        \node[empty] (t34) [right=of t24] {};
        \node[empty] (t44) [right=of t34] {};
        \node[empty] (t54) [right=of t44] {};
        \node[empty] (t64) [right=of t54] {};
        \node[blueCan] (t74) [right=of t64] {};

        \node[blueCan] (t05) at (-.5, -3.5) {};
        \node[empty] (t15) [right=of t05] {};
        \node[empty] (t25) [right=of t15] {};
        \node[empty] (t35) [right=of t25] {};
        \node[empty] (t45) [right=of t35] {};
        \node[empty] (t55) [right=of t45] {};
        \node[empty] (t65) [right=of t55] {};
        \node[redStone] (t75) [right=of t65] {};

        \node[blueCan] (t06) at (0, -4.2) {};
        \node[empty] (t16) [right=of t06] {};
        \node[empty] (t26) [right=of t16] {};
        \node[empty] (t36) [right=of t26] {};
        \node[empty] (t46) [right=of t36] {};
        \node[empty] (t56) [right=of t46] {};
        \node[empty] (t66) [right=of t56] {};
        \node[nix] (t76) [right=of t66] {};

        \node[nix] (t07) at (-.5, -4.9) {};
        \node[redStone] (t17) [right=of t07] {};
        \node[blueStone] (t27) [right=of t17] {};
        \node[redStone] (t37) [right=of t27] {};
        \node[blueCan] (t47) [right=of t37] {};
        \node[nix] (t57) [right=of t47] {};
        \node[nix] (t67) [right=of t57] {};
        \node[blueStone] (t77) [right=of t67] {};
    \end{tikzpicture}\end{center}
    \caption{Gadget Template, showing the state of all vertices on the edges of the gadget.}
    \label{fig:gadgetTemplate}
\end{figure}

The template for the boundary of these tiles is shown in Figure \ref{fig:gadgetTemplate}.  All gadgets will have this configuration along the boundary, with some condsiderations:

\begin{itemize}
    \item Some of the unplayable vertices ($d$) have that status due to neighbors in that gadget.
    \item The remaining unplayable vertices have that status because of neighbors in the adjacent gadgets.  Figure \ref{fig:templateAndNeighbors} shows a template with vertices along the boundary of adjacent gadgets.
    \item Four of the vertices where Blue can play ($e$)---two adjacent in each of the top right and bottom left---will be made either unplayable or already blue depending on the gadget.  Each gadget will have to ensure this.
    \item The remaining four vertices where Blue can play---near the middle of each boundary edge---are the interfaces for this gadget to interact with neighboring gadgets.  Either Blue will be able to play there, or they will be unplayable as needed by the gadget.
\end{itemize}

Since the template has rotational symmetry, all gadgets can be given as specified or rotated 180-degrees as needed in the construction.

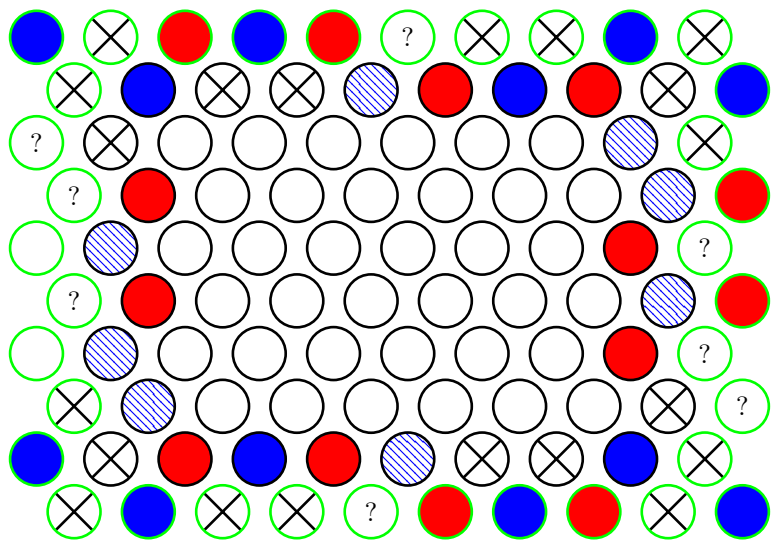
\begin{figure}[h!]
    \begin{center}
    \begin{tikzpicture}[node distance = .25cm, minimum size = .5cm, inner sep = .07cm, ultra thick]
        \node[nix, draw=green] (t0n1) at (-.5, .7) {};
        \node[redStone, draw=green] (t1n1) [right=of t0n1] {};
        \node[blueStone, draw=green] (t2n1) [right=of t1n1] {};
        \node[redStone, draw=green] (t3n1) [right=of t2n1] {};
        \node[empty, draw=green] (t4n1) [right=of t3n1] {?};
        \node[nix, draw=green] (t5n1) [right=of t4n1] {};
        \node[nix, draw=green] (t6n1) [right=of t5n1] {};
        \node[blueStone, draw=green] (t7n1) [right=of t6n1] {};
        \node[nix, draw=green] (t8n1) [right=of t7n1] {};
        \node[blueStone, draw=green] (tn1n1) [left=of t0n1] {};
    
        \node[blueStone] (t00) {};
        \node[nix] (t10) [right=of t00] {};
        \node[nix] (t20) [right=of t10] {};
        \node[blueCan] (t30) [right=of t20] {};
        \node[redStone] (t40) [right=of t30] {};
        \node[blueStone] (t50) [right=of t40] {};
        \node[redStone] (t60) [right=of t50] {};
        \node[nix] (t70) [right=of t60] {};
        \node[blueStone, draw=green] (t80) [right=of t70] {};
        \node[nix, draw=green] (tn10) [left=of t00] {};

        \node[nix] (t01) at (-.5, -.7) {};
        \node[empty] (t11) [right=of t01] {};
        \node[empty] (t21) [right=of t11] {};
        \node[empty] (t31) [right=of t21] {};
        \node[empty] (t41) [right=of t31] {};
        \node[empty] (t51) [right=of t41] {};
        \node[empty] (t61) [right=of t51] {};
        \node[blueCan] (t71) [right=of t61] {};
        \node[nix, draw=green] (t81) [right=of t71] {};
        \node[empty, draw=green] (tn11) [left=of t01] {?};
        
        \node[redStone] (t02) at (0, -1.4) {};
        \node[empty] (t12) [right=of t02] {};
        \node[empty] (t22) [right=of t12] {};
        \node[empty] (t32) [right=of t22] {};
        \node[empty] (t42) [right=of t32] {};
        \node[empty] (t52) [right=of t42] {};
        \node[empty] (t62) [right=of t52] {};
        \node[blueCan] (t72) [right=of t62] {};
        \node[redStone, draw=green] (t82) [right=of t72] {};
        \node[empty, draw=green] (tn12) [left=of t02] {?};

        \node[blueCan] (t03) at (-.5, -2.1) {};
        \node[empty] (t13) [right=of t03] {};
        \node[empty] (t23) [right=of t13] {};
        \node[empty] (t33) [right=of t23] {};
        \node[empty] (t43) [right=of t33] {};
        \node[empty] (t53) [right=of t43] {};
        \node[empty] (t63) [right=of t53] {};
        \node[redStone] (t73) [right=of t63] {};
        \node[empty, draw=green] (t83) [right=of t73] {?};
        \node[empty, draw=green] (tn13) [left=of t03] {};
        
        \node[redStone] (t04) at (0, -2.8) {};
        \node[empty] (t14) [right=of t04] {};
        \node[empty] (t24) [right=of t14] {};
        \node[empty] (t34) [right=of t24] {};
        \node[empty] (t44) [right=of t34] {};
        \node[empty] (t54) [right=of t44] {};
        \node[empty] (t64) [right=of t54] {};
        \node[blueCan] (t74) [right=of t64] {};
        \node[redStone, draw=green] (t84) [right=of t74] {};
        \node[empty, draw=green] (tn14) [left=of t04] {?};

        \node[blueCan] (t05) at (-.5, -3.5) {};
        \node[empty] (t15) [right=of t05] {};
        \node[empty] (t25) [right=of t15] {};
        \node[empty] (t35) [right=of t25] {};
        \node[empty] (t45) [right=of t35] {};
        \node[empty] (t55) [right=of t45] {};
        \node[empty] (t65) [right=of t55] {};
        \node[redStone] (t75) [right=of t65] {};
        \node[empty, draw=green] (t85) [right=of t75] {?};
        \node[empty, draw=green] (tn15) [left=of t05] {};

        \node[blueCan] (t06) at (0, -4.2) {};
        \node[empty] (t16) [right=of t06] {};
        \node[empty] (t26) [right=of t16] {};
        \node[empty] (t36) [right=of t26] {};
        \node[empty] (t46) [right=of t36] {};
        \node[empty] (t56) [right=of t46] {};
        \node[empty] (t66) [right=of t56] {};
        \node[nix] (t76) [right=of t66] {};
        \node[empty, draw=green] (t86) [right=of t76] {?};
        \node[nix, draw=green] (tn16) [left=of t06] {};

        \node[nix] (t07) at (-.5, -4.9) {};
        \node[redStone] (t17) [right=of t07] {};
        \node[blueStone] (t27) [right=of t17] {};
        \node[redStone] (t37) [right=of t27] {};
        \node[blueCan] (t47) [right=of t37] {};
        \node[nix] (t57) [right=of t47] {};
        \node[nix] (t67) [right=of t57] {};
        \node[blueStone] (t77) [right=of t67] {};
        \node[nix, draw=green] (t87) [right=of t77] {};
        \node[blueStone, draw=green] (tn17) [left=of t07] {};

        \node[blueStone, draw=green] (t08) at (0, -5.6) {};
        \node[nix, draw=green] (t18) [right=of t08] {};
        \node[nix, draw=green] (t28) [right=of t18] {};
        \node[empty, draw=green] (t38) [right=of t28] {?};
        \node[redStone, draw=green] (t48) [right=of t38] {};
        \node[blueStone, draw=green] (t58) [right=of t48] {};
        \node[redStone, draw=green] (t68) [right=of t58] {};
        \node[nix, draw=green] (t78) [right=of t68] {};
        \node[blueStone, draw=green] (t88) [right=of t78] {};
        \node[nix, draw=green] (tn18) [left=of t08] {};
    \end{tikzpicture}\end{center}
    \caption{Gadget Template with neighboring circles in green.  Question marks indicate that blue may or may not be able to play at those spots in neighboring gadgets.}
    \label{fig:templateAndNeighbors}
\end{figure}

We describe our gadgets with figures to ease verification of the construction.  Each will show:
\begin{itemize}
    \item Neighboring (T, F) vertex pairs: these are vertices where only Blue can play (except in the variable gadget).  Blue will need to play one vertex in each of these pairs or they will lose the game.  The F vertex will be adjacent to the T vertex in the next pair, so the next T can only be chosen if the prior one was.  It can never benefit Blue to play in an F instead of the paired T vertex, so they will always continue with T vertices if possible.
    \item Neighboring $o$ vertex pairs: these are outputs in a gadget, also only playable by Blue.  The output signal if False if Blue plays in the $o$ vertex on the boundary, and True if they play in the $o$ vertex next to the boundary.
    \item Neighboring $i$ vertex pairs: these act nearly the same as the outputs, except that the True $i$ vertex is the one along the boundary and the False is the one adjacent to it.
\end{itemize}

We first describe the VARIABLE gadget in Figure \ref{fig:variableGadget}.  These are the only main gadgets where Red can play; both players can play at the T and F and only Blue can play at the two $o$.  The player who first plays in each VARIABLE should choose the T-vertex, as that is equivalent to marking the associated variable with their truth value.  After all VARIABLE gadgets have been claimed, Blue will continue to play in gadgets while Red must play elsewhere.  Blue must play twice in each VARIABLE gadget, so they will either choose the T and propagate a True signal, or Red will play on the T and Blue must then play on the F, propagating False from that variable.

\begin{figure}[h!]
    \begin{center}
    \begin{tikzpicture}[node distance = .25cm, minimum size = .5cm, inner sep = .07cm, ultra thick]
        \node[blueStone] (t00) {};
        \node[nix] (t10) [right=of t00] {};
        \node[nix] (t20) [right=of t10] {};
        \node[empty] (t30) [right=of t20] {$o$};
        \node[redStone] (t40) [right=of t30] {};
        \node[blueStone] (t50) [right=of t40] {};
        \node[redStone] (t60) [right=of t50] {};
        \node[nix] (t70) [right=of t60] {};

        \node[nix] (t01) at (-.5, -.7) {};
        \node[empty] (t11) [right=of t01] {};
        \node[blueStone] (t21) [right=of t11] {};
        \node[redStone] (t31) [right=of t21] {};
        \node[empty] (t41) [right=of t31] {$o$};
        \node[empty] (t51) [right=of t41] {};
        \node[empty] (t61) [right=of t51] {};
        \node[blueStone] (t71) [right=of t61] {};
        
        \node[redStone] (t02) at (0, -1.4) {};
        \node[empty] (t12) [right=of t02] {};
        \node[empty] (t22) [right=of t12] {};
        \node[empty] (t32) [right=of t22] {};
        \node[empty] (t42) [right=of t32] {F};
        \node[empty] (t52) [right=of t42] {T};
        \node[empty] (t62) [right=of t52] {};
        \node[nix] (t72) [right=of t62] {};

        \node[nix] (t03) at (-.5, -2.1) {};
        \node[blueStone] (t13) [right=of t03] {};
        \node[redStone] (t23) [right=of t13] {};
        \node[blueStone] (t33) [right=of t23] {};
        \node[empty] (t43) [right=of t33] {};
        \node[empty] (t53) [right=of t43] {};
        \node[empty] (t63) [right=of t53] {};
        \node[redStone] (t73) [right=of t63] {};
        
        \node[redStone] (t04) at (0, -2.8) {};
        \node[empty] (t14) [right=of t04] {};
        \node[empty] (t24) [right=of t14] {};
        \node[redStone] (t34) [right=of t24] {};
        \node[blueStone] (t44) [right=of t34] {};
        \node[redStone] (t54) [right=of t44] {};
        \node[blueStone] (t64) [right=of t54] {};
        \node[nix] (t74) [right=of t64] {};

        \node[nix] (t05) at (-.5, -3.5) {};
        \node[empty] (t15) [right=of t05] {};
        \node[redStone] (t25) [right=of t15] {};
        \node[empty] (t35) [right=of t25] {};
        \node[empty] (t45) [right=of t35] {};
        \node[empty] (t55) [right=of t45] {};
        \node[empty] (t65) [right=of t55] {};
        \node[redStone] (t75) [right=of t65] {};

        \node[blueStone] (t06) at (0, -4.2) {};
        \node[empty] (t16) [right=of t06] {};
        \node[empty] (t26) [right=of t16] {};
        \node[blueStone] (t36) [right=of t26] {};
        \node[redStone] (t46) [right=of t36] {};
        \node[blueStone] (t56) [right=of t46] {}; 
        \node[empty] (t66) [right=of t56] {};
        \node[nix] (t76) [right=of t66] {};

        \node[nix] (t07) at (-.5, -4.9) {};
        \node[redStone] (t17) [right=of t07] {};
        \node[blueStone] (t27) [right=of t17] {};
        \node[redStone] (t37) [right=of t27] {};
        \node[nix] (t47) [right=of t37] {};
        \node[nix] (t57) [right=of t47] {};
        \node[nix] (t67) [right=of t57] {};
        \node[blueStone] (t77) [right=of t67] {};
    \end{tikzpicture}\end{center}
    \caption{VARIABLE Gadget with the output pair ($o$) labeled.  If Blue plays first, they will play on the T, otherwise Red plays there and Blue plays on the F. } 
    \label{fig:variableGadget}
\end{figure}
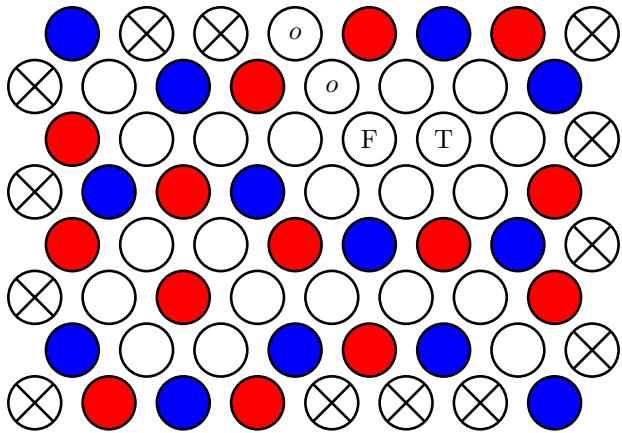

The GOAL gadget, Figure \ref{fig:goal}, is the final gadget that determines who will win the game.  In order for Blue to win, they must be able to play twice in this gadget, meaning once in the $i$ pair and once at G.  (Thus, $i$ must be True.)  If Blue cannot play at G, then Red will make the final move and win the game. 

\begin{figure}[h!]
    \begin{center}
    \begin{tikzpicture}[node distance = .25cm, minimum size = .5cm, inner sep = .07cm, ultra thick]
        \node[blueStone] (t00) {};
        \node[nix] (t10) [right=of t00] {};
        \node[nix] (t20) [right=of t10] {};
        \node[nix] (t30) [right=of t20] {};
        \node[redStone] (t40) [right=of t30] {};
        \node[blueStone] (t50) [right=of t40] {};
        \node[redStone] (t60) [right=of t50] {};
        \node[nix] (t70) [right=of t60] {};

        \node[nix] (t01) at (-.5, -.7) {};
        \node[empty] (t11) [right=of t01] {};
        \node[redStone] (t21) [right=of t11] {};
        \node[blueStone] (t31) [right=of t21] {};
        \node[empty] (t41) [right=of t31] {};
        \node[empty] (t51) [right=of t41] {};
        \node[empty] (t61) [right=of t51] {};
        \node[blueStone] (t71) [right=of t61] {};
        
        \node[redStone] (t02) at (0, -1.4) {};
        \node[empty] (t12) [right=of t02] {};
        \node[empty] (t22) [right=of t12] {};
        \node[redStone] (t32) [right=of t22] {};
        \node[empty] (t42) [right=of t32] {};
        \node[blueStone] (t52) [right=of t42] {};
        \node[empty] (t62) [right=of t52] {};
        \node[nix] (t72) [right=of t62] {};

        \node[nix] (t03) at (-.5, -2.1) {};
        \node[blueStone] (t13) [right=of t03] {};
        \node[empty] (t23) [right=of t13] {};
        \node[empty] (t33) [right=of t23] {};
        \node[blueStone] (t43) [right=of t33] {};
        \node[redStone] (t53) [right=of t43] {};
        \node[empty] (t63) [right=of t53] {};
        \node[redStone] (t73) [right=of t63] {};
        
        \node[redStone] (t04) at (0, -2.8) {};
        \node[empty] (t14) [right=of t04] {};
        \node[redStone] (t24) [right=of t14] {};
        \node[empty] (t34) [right=of t24] {};
        \node[empty] (t44) [right=of t34] {};
        \node[empty] (t54) [right=of t44] {};
        \node[blueStone] (t64) [right=of t54] {};
        \node[nix] (t74) [right=of t64] {};

        \node[nix] (t05) at (-.5, -3.5) {};
        \node[empty] (t15) [right=of t05] {};
        \node[blueStone] (t25) [right=of t15] {};
        \node[empty] (t35) [right=of t25] {};
        \node[empty] (t45) [right=of t35] {G};
        \node[empty] (t55) [right=of t45] {};
        \node[empty] (t65) [right=of t55] {};
        \node[redStone] (t75) [right=of t65] {};

        \node[blueStone] (t06) at (0, -4.2) {};
        \node[empty] (t16) [right=of t06] {};
        \node[empty] (t26) [right=of t16] {};
        \node[empty] (t36) [right=of t26] {$i$};
        \node[redStone] (t46) [right=of t36] {};
        \node[blueStone] (t56) [right=of t46] {};
        \node[empty] (t66) [right=of t56] {};
        \node[nix] (t76) [right=of t66] {};

        \node[nix] (t07) at (-.5, -4.9) {};
        \node[redStone] (t17) [right=of t07] {};
        \node[blueStone] (t27) [right=of t17] {};
        \node[redStone] (t37) [right=of t27] {};
        \node[empty] (t47) [right=of t37] {$i$};
        \node[nix] (t57) [right=of t47] {};
        \node[nix] (t67) [right=of t57] {};
        \node[blueStone] (t77) [right=of t67] {};
    \end{tikzpicture}\end{center}
    \caption{GOAL gadget: In order for Blue to win, they must play at the vertex marked G.}
    \label{fig:goal}
\end{figure}
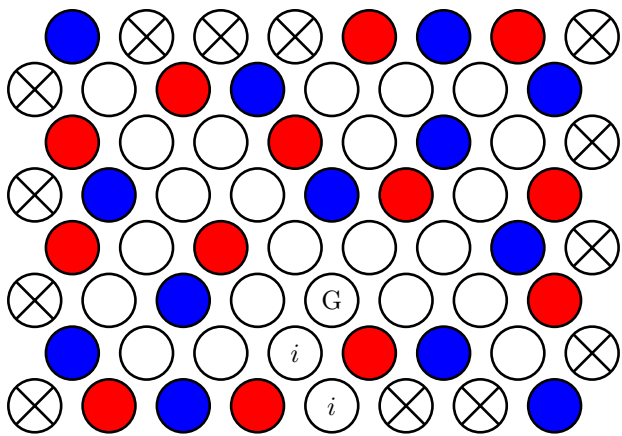

To connect gadgets, we'll need WIRE gadgets.  In these, the output can be True if and only if the input is True.  In the Vertical WIRE, Figure \ref{fig:verticalWireGadget}, Blue must play four times.  This gadget can be reversed to go down without changing the gadget at all.

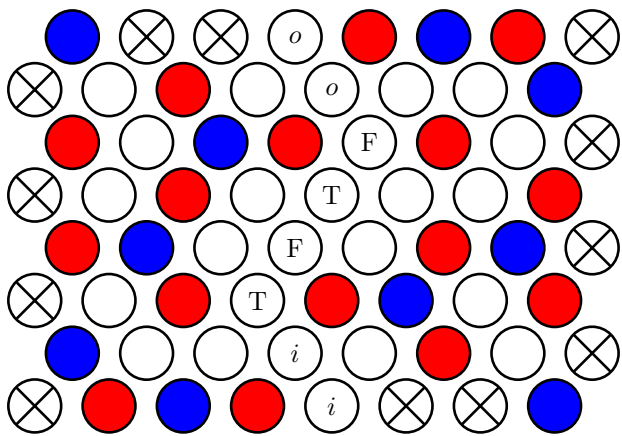
\begin{figure}[h!]
    \begin{center}
    \begin{tikzpicture}[node distance = .25cm, minimum size = .5cm, inner sep = .07cm, ultra thick]
        \node[blueStone] (t00) {};
        \node[nix] (t10) [right=of t00] {};
        \node[nix] (t20) [right=of t10] {};
        \node[empty] (t30) [right=of t20] {$o$};
        \node[redStone] (t40) [right=of t30] {};
        \node[blueStone] (t50) [right=of t40] {};
        \node[redStone] (t60) [right=of t50] {};
        \node[nix] (t70) [right=of t60] {};

        \node[nix] (t01) at (-.5, -.7) {};
        \node[empty] (t11) [right=of t01] {};
        \node[redStone] (t21) [right=of t11] {};
        \node[empty] (t31) [right=of t21] {};
        \node[empty] (t41) [right=of t31] {$o$};
        \node[empty] (t51) [right=of t41] {};
        \node[empty] (t61) [right=of t51] {};
        \node[blueStone] (t71) [right=of t61] {};
        
        \node[redStone] (t02) at (0, -1.4) {};
        \node[empty] (t12) [right=of t02] {};
        \node[blueStone] (t22) [right=of t12] {};
        \node[redStone] (t32) [right=of t22] {};
        \node[empty] (t42) [right=of t32] {F};
        \node[redStone] (t52) [right=of t42] {};
        \node[empty] (t62) [right=of t52] {};
        \node[nix] (t72) [right=of t62] {};

        \node[nix] (t03) at (-.5, -2.1) {};
        \node[empty] (t13) [right=of t03] {};
        \node[redStone] (t23) [right=of t13] {};
        \node[empty] (t33) [right=of t23] {};
        \node[empty] (t43) [right=of t33] {T};
        \node[empty] (t53) [right=of t43] {};
        \node[empty] (t63) [right=of t53] {};
        \node[redStone] (t73) [right=of t63] {};
        
        \node[redStone] (t04) at (0, -2.8) {};
        \node[blueStone] (t14) [right=of t04] {};
        \node[empty] (t24) [right=of t14] {};
        \node[empty] (t34) [right=of t24] {F};
        \node[empty] (t44) [right=of t34] {};
        \node[redStone] (t54) [right=of t44] {};
        \node[blueStone] (t64) [right=of t54] {};
        \node[nix] (t74) [right=of t64] {};

        \node[nix] (t05) at (-.5, -3.5) {};
        \node[empty] (t15) [right=of t05] {};
        \node[redStone] (t25) [right=of t15] {};
        \node[empty] (t35) [right=of t25] {T};
        \node[redStone] (t45) [right=of t35] {};
        \node[blueStone] (t55) [right=of t45] {};
        \node[empty] (t65) [right=of t55] {};
        \node[redStone] (t75) [right=of t65] {};

        \node[blueStone] (t06) at (0, -4.2) {};
        \node[empty] (t16) [right=of t06] {};
        \node[empty] (t26) [right=of t16] {};
        \node[empty] (t36) [right=of t26] {$i$};
        \node[empty] (t46) [right=of t36] {};
        \node[redStone] (t56) [right=of t46] {}; 
        \node[empty] (t66) [right=of t56] {};
        \node[nix] (t76) [right=of t66] {};

        \node[nix] (t07) at (-.5, -4.9) {};
        \node[redStone] (t17) [right=of t07] {};
        \node[blueStone] (t27) [right=of t17] {};
        \node[redStone] (t37) [right=of t27] {};
        \node[empty] (t47) [right=of t37] {$i$};
        \node[nix] (t57) [right=of t47] {};
        \node[nix] (t67) [right=of t57] {};
        \node[blueStone] (t77) [right=of t67] {};
    \end{tikzpicture}\end{center}
    \caption{Vertical WIRE Gadget with the input pair ($i$) and output pair ($o$) labeled.}
    \label{fig:verticalWireGadget}
\end{figure}

In the Horizontal WIRE gadget, Figure \ref{fig:horizontalWireGadget}, Blue must play five times.  This can also be reversed without any changes.

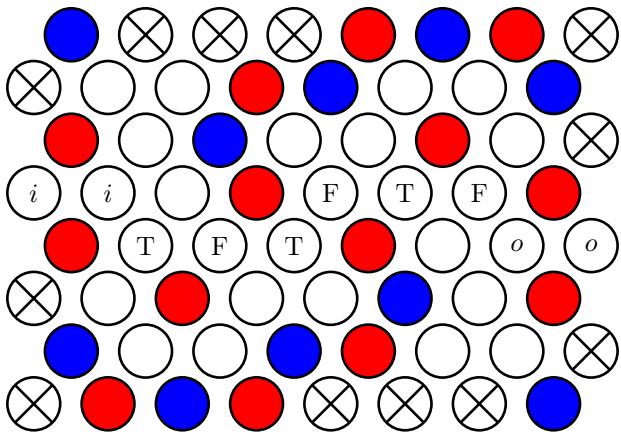
\begin{figure}[h!]
    \begin{center}
    \begin{tikzpicture}[node distance = .25cm, minimum size = .5cm, inner sep = .07cm, ultra thick]
        \node[blueStone] (t00) {};
        \node[nix] (t10) [right=of t00] {};
        \node[nix] (t20) [right=of t10] {};
        \node[nix] (t30) [right=of t20] {};
        \node[redStone] (t40) [right=of t30] {};
        \node[blueStone] (t50) [right=of t40] {};
        \node[redStone] (t60) [right=of t50] {};
        \node[nix] (t70) [right=of t60] {};

        \node[nix] (t01) at (-.5, -.7) {};
        \node[empty] (t11) [right=of t01] {};
        \node[empty] (t21) [right=of t11] {};
        \node[redStone] (t31) [right=of t21] {};
        \node[blueStone] (t41) [right=of t31] {};
        \node[empty] (t51) [right=of t41] {};
        \node[empty] (t61) [right=of t51] {};
        \node[blueStone] (t71) [right=of t61] {};
        
        \node[redStone] (t02) at (0, -1.4) {};
        \node[empty] (t12) [right=of t02] {};
        \node[blueStone] (t22) [right=of t12] {};
        \node[empty] (t32) [right=of t22] {};
        \node[empty] (t42) [right=of t32] {};
        \node[redStone] (t52) [right=of t42] {};
        \node[empty] (t62) [right=of t52] {};
        \node[nix] (t72) [right=of t62] {};

        \node[empty] (t03) at (-.5, -2.1) {$i$};
        \node[empty] (t13) [right=of t03] {$i$};
        \node[empty] (t23) [right=of t13] {};
        \node[redStone] (t33) [right=of t23] {};
        \node[empty] (t43) [right=of t33] {F};
        \node[empty] (t53) [right=of t43] {T};
        \node[empty] (t63) [right=of t53] {F};
        \node[redStone] (t73) [right=of t63] {};
        
        \node[redStone] (t04) at (0, -2.8) {};
        \node[empty] (t14) [right=of t04] {T};
        \node[empty] (t24) [right=of t14] {F};
        \node[empty] (t34) [right=of t24] {T};
        \node[redStone] (t44) [right=of t34] {};
        \node[empty] (t54) [right=of t44] {};
        \node[empty] (t64) [right=of t54] {$o$};
        \node[empty] (t74) [right=of t64] {$o$};

        \node[nix] (t05) at (-.5, -3.5) {};
        \node[empty] (t15) [right=of t05] {};
        \node[redStone] (t25) [right=of t15] {};
        \node[empty] (t35) [right=of t25] {};
        \node[empty] (t45) [right=of t35] {};
        \node[blueStone] (t55) [right=of t45] {};
        \node[empty] (t65) [right=of t55] {};
        \node[redStone] (t75) [right=of t65] {};

        \node[blueStone] (t06) at (0, -4.2) {};
        \node[empty] (t16) [right=of t06] {};
        \node[empty] (t26) [right=of t16] {};
        \node[blueStone] (t36) [right=of t26] {};
        \node[redStone] (t46) [right=of t36] {};
        \node[empty] (t56) [right=of t46] {}; 
        \node[empty] (t66) [right=of t56] {};
        \node[nix] (t76) [right=of t66] {};

        \node[nix] (t07) at (-.5, -4.9) {};
        \node[redStone] (t17) [right=of t07] {};
        \node[blueStone] (t27) [right=of t17] {};
        \node[redStone] (t37) [right=of t27] {};
        \node[nix] (t47) [right=of t37] {};
        \node[nix] (t57) [right=of t47] {};
        \node[nix] (t67) [right=of t57] {};
        \node[blueStone] (t77) [right=of t67] {};
    \end{tikzpicture}\end{center}
    \caption{Horizontal WIRE Gadget with the input pair ($i$) and output pair ($o$) labeled.}
    \label{fig:horizontalWireGadget}
\end{figure}

In the Right Turn WIRE gadget, Figure \ref{fig:rightTurnWireGadget}, Blue must play three times.  This can be reversed without change, and can be rotated 180-degrees to instead connect the top and left sides in either direction.

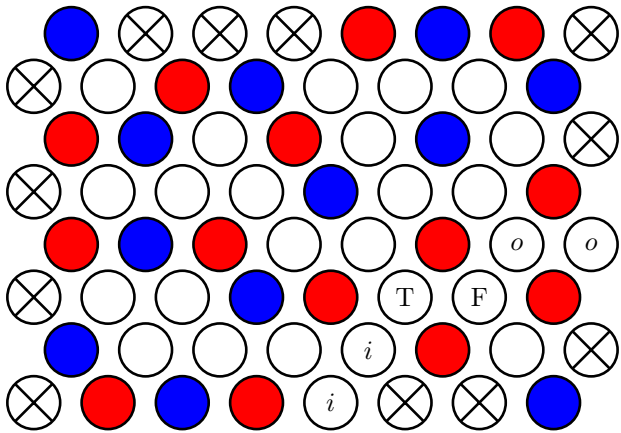
\begin{figure}[h!]
    \begin{center}
    \begin{tikzpicture}[node distance = .25cm, minimum size = .5cm, inner sep = .07cm, ultra thick]
        \node[blueStone] (t00) {};
        \node[nix] (t10) [right=of t00] {};
        \node[nix] (t20) [right=of t10] {};
        \node[nix] (t30) [right=of t20] {};
        \node[redStone] (t40) [right=of t30] {};
        \node[blueStone] (t50) [right=of t40] {};
        \node[redStone] (t60) [right=of t50] {};
        \node[nix] (t70) [right=of t60] {};

        \node[nix] (t01) at (-.5, -.7) {};
        \node[empty] (t11) [right=of t01] {};
        \node[redStone] (t21) [right=of t11] {};
        \node[blueStone] (t31) [right=of t21] {};
        \node[empty] (t41) [right=of t31] {};
        \node[empty] (t51) [right=of t41] {};
        \node[empty] (t61) [right=of t51] {};
        \node[blueStone] (t71) [right=of t61] {};
        
        \node[redStone] (t02) at (0, -1.4) {};
        \node[blueStone] (t12) [right=of t02] {};
        \node[empty] (t22) [right=of t12] {};
        \node[redStone] (t32) [right=of t22] {};
        \node[empty] (t42) [right=of t32] {};
        \node[blueStone] (t52) [right=of t42] {};
        \node[empty] (t62) [right=of t52] {};
        \node[nix] (t72) [right=of t62] {};

        \node[nix] (t03) at (-.5, -2.1) {};
        \node[empty] (t13) [right=of t03] {};
        \node[empty] (t23) [right=of t13] {};
        \node[empty] (t33) [right=of t23] {};
        \node[blueStone] (t43) [right=of t33] {};
        \node[empty] (t53) [right=of t43] {};
        \node[empty] (t63) [right=of t53] {};
        \node[redStone] (t73) [right=of t63] {};
        
        \node[redStone] (t04) at (0, -2.8) {};
        \node[blueStone] (t14) [right=of t04] {};
        \node[redStone] (t24) [right=of t14] {};
        \node[empty] (t34) [right=of t24] {};
        \node[empty] (t44) [right=of t34] {};
        \node[redStone] (t54) [right=of t44] {};
        \node[empty] (t64) [right=of t54] {$o$};
        \node[empty] (t74) [right=of t64] {$o$};

        \node[nix] (t05) at (-.5, -3.5) {};
        \node[empty] (t15) [right=of t05] {};
        \node[empty] (t25) [right=of t15] {};
        \node[blueStone] (t35) [right=of t25] {};
        \node[redStone] (t45) [right=of t35] {};
        \node[empty] (t55) [right=of t45] {T};
        \node[empty] (t65) [right=of t55] {F};
        \node[redStone] (t75) [right=of t65] {};

        \node[blueStone] (t06) at (0, -4.2) {};
        \node[empty] (t16) [right=of t06] {};
        \node[empty] (t26) [right=of t16] {};
        \node[empty] (t36) [right=of t26] {};
        \node[empty] (t46) [right=of t36] {$i$};
        \node[redStone] (t56) [right=of t46] {}; 
        \node[empty] (t66) [right=of t56] {};
        \node[nix] (t76) [right=of t66] {};

        \node[nix] (t07) at (-.5, -4.9) {};
        \node[redStone] (t17) [right=of t07] {};
        \node[blueStone] (t27) [right=of t17] {};
        \node[redStone] (t37) [right=of t27] {};
        \node[empty] (t47) [right=of t37] {$i$};
        \node[nix] (t57) [right=of t47] {};
        \node[nix] (t67) [right=of t57] {};
        \node[blueStone] (t77) [right=of t67] {};
    \end{tikzpicture}\end{center}
    \caption{Right Turn WIRE Gadget with the input pair ($i$) and output pair ($o$) labeled.}
    \label{fig:rightTurnWireGadget}
\end{figure}

In the Left Turn WIRE gadget, Figure \ref{fig:leftTurnWireGadget}, Blue must play four times.  This gadget is reversible without modifications, and can be rotated 180 degrees to connect the top and right hand sides in either direction.

\begin{figure}[h!]
    \begin{center}
    \begin{tikzpicture}[node distance = .25cm, minimum size = .5cm, inner sep = .07cm, ultra thick]
        \node[blueStone] (t00) {};
        \node[nix] (t10) [right=of t00] {};
        \node[nix] (t20) [right=of t10] {};
        \node[nix] (t30) [right=of t20] {};
        \node[redStone] (t40) [right=of t30] {};
        \node[blueStone] (t50) [right=of t40] {};
        \node[redStone] (t60) [right=of t50] {};
        \node[nix] (t70) [right=of t60] {};

        \node[nix] (t01) at (-.5, -.7) {};
        \node[empty] (t11) [right=of t01] {};
        \node[blueStone] (t21) [right=of t11] {};
        \node[empty] (t31) [right=of t21] {};
        \node[empty] (t41) [right=of t31] {};
        \node[empty] (t51) [right=of t41] {};
        \node[empty] (t61) [right=of t51] {};
        \node[blueStone] (t71) [right=of t61] {};
        
        \node[redStone] (t02) at (0, -1.4) {};
        \node[empty] (t12) [right=of t02] {};
        \node[redStone] (t22) [right=of t12] {};
        \node[empty] (t32) [right=of t22] {};
        \node[blueStone] (t42) [right=of t32] {};
        \node[redStone] (t52) [right=of t42] {};
        \node[empty] (t62) [right=of t52] {};
        \node[nix] (t72) [right=of t62] {};

        \node[empty] (t03) at (-.5, -2.1) {$o$};
        \node[empty] (t13) [right=of t03] {$o$};
        \node[empty] (t23) [right=of t13] {F};
        \node[empty] (t33) [right=of t23] {T};
        \node[empty] (t43) [right=of t33] {};
        \node[empty] (t53) [right=of t43] {};
        \node[empty] (t63) [right=of t53] {};
        \node[redStone] (t73) [right=of t63] {};
        
        \node[redStone] (t04) at (0, -2.8) {};
        \node[empty] (t14) [right=of t04] {};
        \node[empty] (t24) [right=of t14] {};
        \node[empty] (t34) [right=of t24] {F};
        \node[redStone] (t44) [right=of t34] {};
        \node[empty] (t54) [right=of t44] {};
        \node[blueStone] (t64) [right=of t54] {};
        \node[nix] (t74) [right=of t64] {};

        \node[nix] (t05) at (-.5, -3.5) {};
        \node[empty] (t15) [right=of t05] {};
        \node[blueStone] (t25) [right=of t15] {};
        \node[redStone] (t35) [right=of t25] {};
        \node[empty] (t45) [right=of t35] {T};
        \node[empty] (t55) [right=of t45] {};
        \node[empty] (t65) [right=of t55] {};
        \node[redStone] (t75) [right=of t65] {};

        \node[blueStone] (t06) at (0, -4.2) {};
        \node[empty] (t16) [right=of t06] {};
        \node[empty] (t26) [right=of t16] {};
        \node[empty] (t36) [right=of t26] {$i$};
        \node[redStone] (t46) [right=of t36] {};
        \node[blueStone] (t56) [right=of t46] {}; 
        \node[empty] (t66) [right=of t56] {};
        \node[nix] (t76) [right=of t66] {};

        \node[nix] (t07) at (-.5, -4.9) {};
        \node[redStone] (t17) [right=of t07] {};
        \node[blueStone] (t27) [right=of t17] {};
        \node[redStone] (t37) [right=of t27] {};
        \node[empty] (t47) [right=of t37] {$i$};
        \node[nix] (t57) [right=of t47] {};
        \node[nix] (t67) [right=of t57] {};
        \node[blueStone] (t77) [right=of t67] {};
    \end{tikzpicture}\end{center}
    \caption{Left Turn WIRE Gadget with the input pair ($i$) and output pair ($o$) labeled.}
    \label{fig:leftTurnWireGadget}
\end{figure}
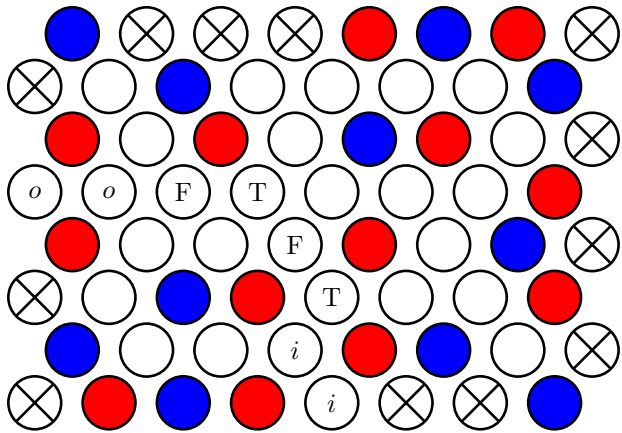

The AND gadget, Figure \ref{fig:andGadget}, takes two inputs.  The T-vertex in the center can be chosen (and the output is True) only if both inputs are True.  Blue must play six times on each of these gadgets.
\begin{figure}[h!]
    \begin{center}
    \begin{tikzpicture}[node distance = .25cm, minimum size = .5cm, inner sep = .07cm, ultra thick]
        \node[blueStone] (t00) {};
        \node[nix] (t10) [right=of t00] {};
        \node[nix] (t20) [right=of t10] {};
        \node[empty] (t30) [right=of t20] {$o$};
        \node[redStone] (t40) [right=of t30] {};
        \node[blueStone] (t50) [right=of t40] {};
        \node[redStone] (t60) [right=of t50] {};
        \node[nix] (t70) [right=of t60] {};

        \node[nix] (t01) at (-.5, -.7) {};
        \node[empty] (t11) [right=of t01] {};
        \node[empty] (t21) [right=of t11] {F};
        \node[empty] (t31) [right=of t21] {$o$};
        \node[empty] (t41) [right=of t31] {x};
        \node[empty] (t51) [right=of t41] {};
        \node[empty] (t61) [right=of t51] {};
        \node[blueStone] (t71) [right=of t61] {};
        
        \node[redStone] (t02) at (0, -1.4) {};
        \node[empty] (t12) [right=of t02] {T};
        \node[redStone] (t22) [right=of t12] {};
        \node[empty] (t32) [right=of t22] {F};
        \node[empty] (t42) [right=of t32] {T};
        \node[empty] (t52) [right=of t42] {F};
        \node[empty] (t62) [right=of t52] {};
        \node[nix] (t72) [right=of t62] {};

        \node[empty] (t03) at (-.5, -2.1) {$i_1$};
        \node[empty] (t13) [right=of t03] {$i_1$};
        \node[empty] (t23) [right=of t13] {};
        \node[empty] (t33) [right=of t23] {};
        \node[empty] (t43) [right=of t33] {};
        \node[redStone] (t53) [right=of t43] {};
        \node[empty] (t63) [right=of t53] {T};
        \node[redStone] (t73) [right=of t63] {};
        
        \node[redStone] (t04) at (0, -2.8) {};
        \node[empty] (t14) [right=of t04] {};
        \node[blueStone] (t24) [right=of t14] {};
        \node[redStone] (t34) [right=of t24] {};
        \node[blueStone] (t44) [right=of t34] {};
        \node[empty] (t54) [right=of t44] {};
        \node[empty] (t64) [right=of t54] {$i_2$};
        \node[empty] (t74) [right=of t64] {$i_2$};

        \node[nix] (t05) at (-.5, -3.5) {};
        \node[empty] (t15) [right=of t05] {};
        \node[redStone] (t25) [right=of t15] {};
        \node[empty] (t35) [right=of t25] {};
        \node[empty] (t45) [right=of t35] {};
        \node[redStone] (t55) [right=of t45] {};
        \node[empty] (t65) [right=of t55] {};
        \node[redStone] (t75) [right=of t65] {};

        \node[blueStone] (t06) at (0, -4.2) {};
        \node[empty] (t16) [right=of t06] {};
        \node[empty] (t26) [right=of t16] {};
        \node[blueStone] (t36) [right=of t26] {};
        \node[empty] (t46) [right=of t36] {};
        \node[blueStone] (t56) [right=of t46] {};
        \node[empty] (t66) [right=of t56] {};
        \node[nix] (t76) [right=of t66] {};

        \node[nix] (t07) at (-.5, -4.9) {};
        \node[redStone] (t17) [right=of t07] {};
        \node[blueStone] (t27) [right=of t17] {};
        \node[redStone] (t37) [right=of t27] {};
        \node[nix] (t47) [right=of t37] {};
        \node[nix] (t57) [right=of t47] {};
        \node[nix] (t67) [right=of t57] {};
        \node[blueStone] (t77) [right=of t67] {};
    \end{tikzpicture}\end{center}
    \caption{AND Gadget with the input pairs ($i_1$ and $i_2$), output pair ($o$), and playable pathways labeled.  Blue cannot win if they play at the `x' because it will prevent them from playing at two other locations.}
    \label{fig:andGadget}
\end{figure}
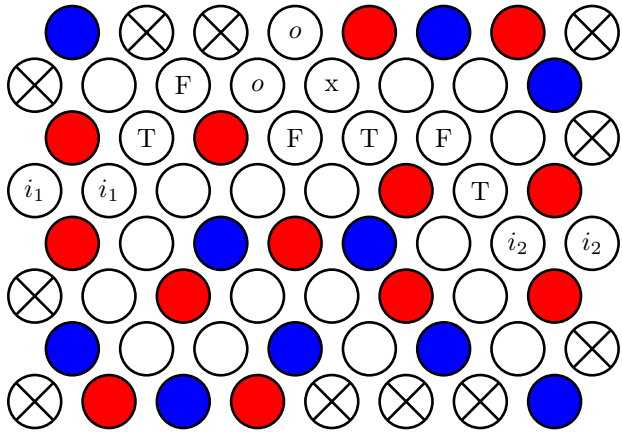

The OR gadget, Figure \ref{fig:orGadget}, also has two inputs and one output.  In the middle is a T-T-F triple instead of a pair.  If either of the inputs is True, then one of the two T's can be chosen in that triple and the output can also be True.  Blue must play six times in the gadget.

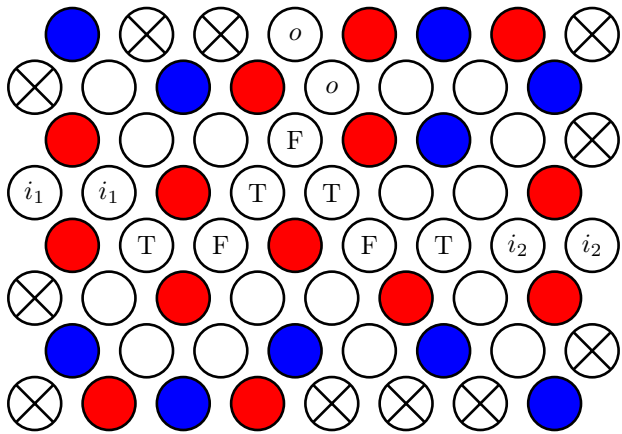
\begin{figure}[h!]
    \begin{center}
    \begin{tikzpicture}[node distance = .25cm, minimum size = .5cm, inner sep = .07cm, ultra thick]
        \node[blueStone] (t00) {};
        \node[nix] (t10) [right=of t00] {};
        \node[nix] (t20) [right=of t10] {};
        \node[empty] (t30) [right=of t20] {$o$};
        \node[redStone] (t40) [right=of t30] {};
        \node[blueStone] (t50) [right=of t40] {};
        \node[redStone] (t60) [right=of t50] {};
        \node[nix] (t70) [right=of t60] {};

        \node[nix] (t01) at (-.5, -.7) {};
        \node[empty] (t11) [right=of t01] {};
        \node[blueStone] (t21) [right=of t11] {};
        \node[redStone] (t31) [right=of t21] {};
        \node[empty] (t41) [right=of t31] {$o$};
        \node[empty] (t51) [right=of t41] {};
        \node[empty] (t61) [right=of t51] {};
        \node[blueStone] (t71) [right=of t61] {};
        
        \node[redStone] (t02) at (0, -1.4) {};
        \node[empty] (t12) [right=of t02] {};
        \node[empty] (t22) [right=of t12] {};
        \node[empty] (t32) [right=of t22] {F};
        \node[redStone] (t42) [right=of t32] {};
        \node[blueStone] (t52) [right=of t42] {};
        \node[empty] (t62) [right=of t52] {};
        \node[nix] (t72) [right=of t62] {};

        \node[empty] (t03) at (-.5, -2.1) {$i_1$};
        \node[empty] (t13) [right=of t03] {$i_1$};
        \node[redStone] (t23) [right=of t13] {};
        \node[empty] (t33) [right=of t23] {T};
        \node[empty] (t43) [right=of t33] {T};
        \node[empty] (t53) [right=of t43] {};
        \node[empty] (t63) [right=of t53] {};
        \node[redStone] (t73) [right=of t63] {};
        
        \node[redStone] (t04) at (0, -2.8) {};
        \node[empty] (t14) [right=of t04] {T};
        \node[empty] (t24) [right=of t14] {F};
        \node[redStone] (t34) [right=of t24] {};
        \node[empty] (t44) [right=of t34] {F};
        \node[empty] (t54) [right=of t44] {T};
        \node[empty] (t64) [right=of t54] {$i_2$};
        \node[empty] (t74) [right=of t64] {$i_2$};

        \node[nix] (t05) at (-.5, -3.5) {};
        \node[empty] (t15) [right=of t05] {};
        \node[redStone] (t25) [right=of t15] {};
        \node[empty] (t35) [right=of t25] {};
        \node[empty] (t45) [right=of t35] {};
        \node[redStone] (t55) [right=of t45] {};
        \node[empty] (t65) [right=of t55] {};
        \node[redStone] (t75) [right=of t65] {};

        \node[blueStone] (t06) at (0, -4.2) {};
        \node[empty] (t16) [right=of t06] {};
        \node[empty] (t26) [right=of t16] {};
        \node[blueStone] (t36) [right=of t26] {};
        \node[empty] (t46) [right=of t36] {};
        \node[blueStone] (t56) [right=of t46] {};
        \node[empty] (t66) [right=of t56] {};
        \node[nix] (t76) [right=of t66] {};

        \node[nix] (t07) at (-.5, -4.9) {};
        \node[redStone] (t17) [right=of t07] {};
        \node[blueStone] (t27) [right=of t17] {};
        \node[redStone] (t37) [right=of t27] {};
        \node[nix] (t47) [right=of t37] {};
        \node[nix] (t57) [right=of t47] {};
        \node[nix] (t67) [right=of t57] {};
        \node[blueStone] (t77) [right=of t67] {};
    \end{tikzpicture}\end{center}
    \caption{OR Gadget with the input pairs ($i_1$ and $i_2$), output pair ($o$), and playable pathways labeled.}
    \label{fig:orGadget}
\end{figure}

The CHOICE gadget, Figure \ref{fig:choiceGadget}, has one input and two outputs and has two T-F pairs with adjacent T-vertices.  If the input is True, then only one of them can be chosen to be True, so Blue will have to choose.  If the input is False, then both outputs must be False.  Blue must play seven times in each of these gadgets.

\begin{figure}[h!]
    \begin{center}
    \begin{tikzpicture}[node distance = .25cm, minimum size = .5cm, inner sep = .07cm, ultra thick]
        \node[blueStone] (t00) {};
        \node[nix] (t10) [right=of t00] {};
        \node[nix] (t20) [right=of t10] {};
        \node[nix] (t30) [right=of t20] {};
        \node[redStone] (t40) [right=of t30] {};
        \node[blueStone] (t50) [right=of t40] {};
        \node[redStone] (t60) [right=of t50] {};
        \node[nix] (t70) [right=of t60] {};

        \node[nix] (t01) at (-.5, -.7) {};
        \node[empty] (t11) [right=of t01] {};
        \node[redStone] (t21) [right=of t11] {};
        \node[empty] (t31) [right=of t21] {};
        \node[blueStone] (t41) [right=of t31] {};
        \node[empty] (t51) [right=of t41] {};
        \node[empty] (t61) [right=of t51] {};
        \node[blueStone] (t71) [right=of t61] {};
        
        \node[redStone] (t02) at (0, -1.4) {};
        \node[empty] (t12) [right=of t02] {F};
        \node[empty] (t22) [right=of t12] {T};
        \node[redStone] (t32) [right=of t22] {};
        \node[empty] (t42) [right=of t32] {};
        \node[redStone] (t52) [right=of t42] {};
        \node[empty] (t62) [right=of t52] {};
        \node[nix] (t72) [right=of t62] {};

        \node[empty] (t03) at (-.5, -2.1) {$o_1$};
        \node[empty] (t13) [right=of t03] {$o_1$};
        \node[redStone] (t23) [right=of t13] {};
        \node[empty] (t33) [right=of t23] {F};
        \node[empty] (t43) [right=of t33] {T};
        \node[empty] (t53) [right=of t43] {T};
        \node[empty] (t63) [right=of t53] {F};
        \node[redStone] (t73) [right=of t63] {};
        
        \node[redStone] (t04) at (0, -2.8) {};
        \node[empty] (t14) [right=of t04] {};
        \node[empty] (t24) [right=of t14] {};
        \node[redStone] (t34) [right=of t24] {};
        \node[empty] (t44) [right=of t34] {F};
        \node[redStone] (t54) [right=of t44] {};
        \node[empty] (t64) [right=of t54] {$o_2$};
        \node[empty] (t74) [right=of t64] {$o_2$};

        \node[nix] (t05) at (-.5, -3.5) {};
        \node[empty] (t15) [right=of t05] {};
        \node[blueStone] (t25) [right=of t15] {};
        \node[empty] (t35) [right=of t25] {};
        \node[empty] (t45) [right=of t35] {T};
        \node[empty] (t55) [right=of t45] {};
        \node[empty] (t65) [right=of t55] {};
        \node[redStone] (t75) [right=of t65] {};

        \node[blueStone] (t06) at (0, -4.2) {};
        \node[empty] (t16) [right=of t06] {};
        \node[empty] (t26) [right=of t16] {};
        \node[empty] (t36) [right=of t26] {$i$};
        \node[redStone] (t46) [right=of t36] {};
        \node[blueStone] (t56) [right=of t46] {}; 
        \node[empty] (t66) [right=of t56] {};
        \node[nix] (t76) [right=of t66] {};

        \node[nix] (t07) at (-.5, -4.9) {};
        \node[redStone] (t17) [right=of t07] {};
        \node[blueStone] (t27) [right=of t17] {};
        \node[redStone] (t37) [right=of t27] {};
        \node[empty] (t47) [right=of t37] {$i$};
        \node[nix] (t57) [right=of t47] {};
        \node[nix] (t67) [right=of t57] {};
        \node[blueStone] (t77) [right=of t67] {};
    \end{tikzpicture}\end{center}
    \caption{CHOICE Gadget with the input pair ($i$), output pairs ($o_1, o_2$), and playable pathways labeled.}
    \label{fig:choiceGadget}
\end{figure}
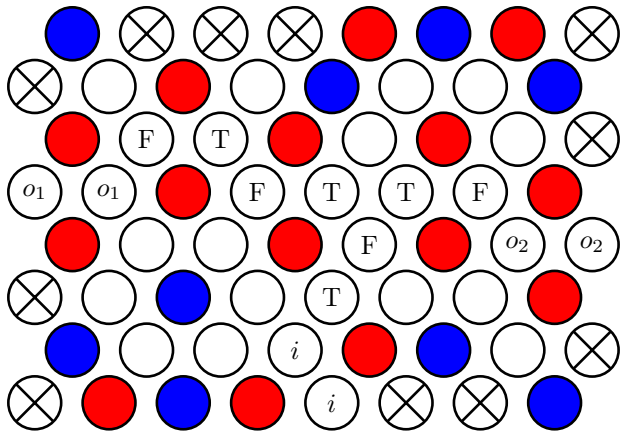

The SPLIT gadget, Figure \ref{fig:splitGadget}, also has one input and two outputs.  Both outputs can be True exactly when the input is True.  Blue must play six times in each of these gadgets.

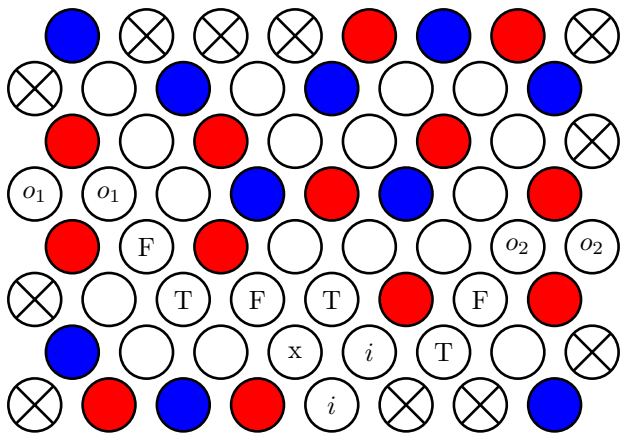
\begin{figure}[h!]
    \begin{center}
    \begin{tikzpicture}[node distance = .25cm, minimum size = .5cm, inner sep = .07cm, ultra thick]
        \node[blueStone] (t00) {};
        \node[nix] (t10) [right=of t00] {};
        \node[nix] (t20) [right=of t10] {};
        \node[nix] (t30) [right=of t20] {};
        \node[redStone] (t40) [right=of t30] {};
        \node[blueStone] (t50) [right=of t40] {};
        \node[redStone] (t60) [right=of t50] {};
        \node[nix] (t70) [right=of t60] {};

        \node[nix] (t01) at (-.5, -.7) {};
        \node[empty] (t11) [right=of t01] {};
        \node[blueStone] (t21) [right=of t11] {};
        \node[empty] (t31) [right=of t21] {};
        \node[blueStone] (t41) [right=of t31] {};
        \node[empty] (t51) [right=of t41] {};
        \node[empty] (t61) [right=of t51] {};
        \node[blueStone] (t71) [right=of t61] {};
        
        \node[redStone] (t02) at (0, -1.4) {};
        \node[empty] (t12) [right=of t02] {};
        \node[redStone] (t22) [right=of t12] {};
        \node[empty] (t32) [right=of t22] {};
        \node[empty] (t42) [right=of t32] {};
        \node[redStone] (t52) [right=of t42] {};
        \node[empty] (t62) [right=of t52] {};
        \node[nix] (t72) [right=of t62] {};

        \node[empty] (t03) at (-.5, -2.1) {$o_1$};
        \node[empty] (t13) [right=of t03] {$o_1$};
        \node[empty] (t23) [right=of t13] {};
        \node[blueStone] (t33) [right=of t23] {};
        \node[redStone] (t43) [right=of t33] {};
        \node[blueStone] (t53) [right=of t43] {};
        \node[empty] (t63) [right=of t53] {};
        \node[redStone] (t73) [right=of t63] {};
        
        \node[redStone] (t04) at (0, -2.8) {};
        \node[empty] (t14) [right=of t04] {F};
        \node[redStone] (t24) [right=of t14] {};
        \node[empty] (t34) [right=of t24] {};
        \node[empty] (t44) [right=of t34] {};
        \node[empty] (t54) [right=of t44] {};
        \node[empty] (t64) [right=of t54] {$o_2$};
        \node[empty] (t74) [right=of t64] {$o_2$};

        \node[nix] (t05) at (-.5, -3.5) {};
        \node[empty] (t15) [right=of t05] {};
        \node[empty] (t25) [right=of t15] {T};
        \node[empty] (t35) [right=of t25] {F};
        \node[empty] (t45) [right=of t35] {T};
        \node[redStone] (t55) [right=of t45] {};
        \node[empty] (t65) [right=of t55] {F};
        \node[redStone] (t75) [right=of t65] {};

        \node[blueStone] (t06) at (0, -4.2) {};
        \node[empty] (t16) [right=of t06] {};
        \node[empty] (t26) [right=of t16] {};
        \node[empty] (t36) [right=of t26] {x};
        \node[empty] (t46) [right=of t36] {$i$};
        \node[empty] (t56) [right=of t46] {T}; 
        \node[empty] (t66) [right=of t56] {};
        \node[nix] (t76) [right=of t66] {};

        \node[nix] (t07) at (-.5, -4.9) {};
        \node[redStone] (t17) [right=of t07] {};
        \node[blueStone] (t27) [right=of t17] {};
        \node[redStone] (t37) [right=of t27] {};
        \node[empty] (t47) [right=of t37] {$i$};
        \node[nix] (t57) [right=of t47] {};
        \node[nix] (t67) [right=of t57] {};
        \node[blueStone] (t77) [right=of t67] {};
    \end{tikzpicture}\end{center}
    \caption{SPLIT Gadget with the input pair ($i$), output pairs ($o_1, o_2$), and playable pathways labeled.  Blue can't win if they play at the `x', because it prevents them from making two other plays.}
    \label{fig:splitGadget}
\end{figure}

In addition to the above circuity, we also need to have spaces for Red to play while Blue is filling in the rest of the circuitry.  Let $x$ be the number of VARIABLE gadgets, $y$ be the number of vertical WIREs, $z$ be the number of horizontal WIREs, $w$ be the number of right turn WIREs, $v$ be the number of left turn WIREs, $u$ be the number of ANDs, $t$ be the number of ORs, $s$ be the number of CHOICEs, and $r$ be the number of SPLITs.  Then in total, the number of spaces we need for Red to play is:

$$1x + 4y + 5z + 3w + 4v + 6u + 6t + 7s + 6r + 1.$$

(The final 1 is for Red's move on the GOAL; recall that $x$ will be even because we have an even number of variables.)  In order to get those spaces, we use REDPLAYS gadgets, shown in Figure \ref{fig:redPlays}.  Each of these gives Red eight moves.  If the above sum comes to anything other than a multiple of 8, some of the red-playable vertices can begin colored red.

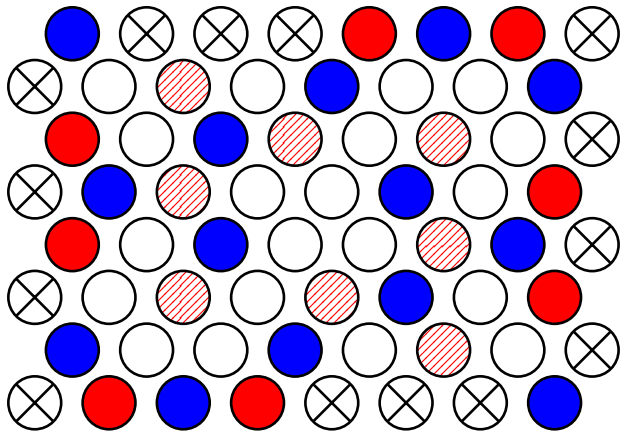
\begin{figure}[h!]
    \begin{center}
    \begin{tikzpicture}[node distance = .25cm, minimum size = .5cm, inner sep = .07cm, ultra thick]
        \node[blueStone] (t00) {};
        \node[nix] (t10) [right=of t00] {};
        \node[nix] (t20) [right=of t10] {};
        \node[nix] (t30) [right=of t20] {};
        \node[redStone] (t40) [right=of t30] {};
        \node[blueStone] (t50) [right=of t40] {};
        \node[redStone] (t60) [right=of t50] {};
        \node[nix] (t70) [right=of t60] {};

        \node[nix] (t01) at (-.5, -.7) {};
        \node[empty] (t11) [right=of t01] {};
        \node[redCan] (t21) [right=of t11] {};
        \node[empty] (t31) [right=of t21] {};
        \node[blueStone] (t41) [right=of t31] {};
        \node[empty] (t51) [right=of t41] {};
        \node[empty] (t61) [right=of t51] {};
        \node[blueStone] (t71) [right=of t61] {};
        
        \node[redStone] (t02) at (0, -1.4) {};
        \node[empty] (t12) [right=of t02] {};
        \node[blueStone] (t22) [right=of t12] {};
        \node[redCan] (t32) [right=of t22] {};
        \node[empty] (t42) [right=of t32] {};
        \node[redCan] (t52) [right=of t42] {};
        \node[empty] (t62) [right=of t52] {};
        \node[nix] (t72) [right=of t62] {};

        \node[nix] (t03) at (-.5, -2.1) {};
        \node[blueStone] (t13) [right=of t03] {};
        \node[redCan] (t23) [right=of t13] {};
        \node[empty] (t33) [right=of t23] {};
        \node[empty] (t43) [right=of t33] {};
        \node[blueStone] (t53) [right=of t43] {};
        \node[empty] (t63) [right=of t53] {};
        \node[redStone] (t73) [right=of t63] {};
        
        \node[redStone] (t04) at (0, -2.8) {};
        \node[empty] (t14) [right=of t04] {};
        \node[blueStone] (t24) [right=of t14] {};
        \node[empty] (t34) [right=of t24] {};
        \node[empty] (t44) [right=of t34] {};
        \node[redCan] (t54) [right=of t44] {};
        \node[blueStone] (t64) [right=of t54] {};
        \node[nix] (t74) [right=of t64] {};

        \node[nix] (t05) at (-.5, -3.5) {};
        \node[empty] (t15) [right=of t05] {};
        \node[redCan] (t25) [right=of t15] {};
        \node[empty] (t35) [right=of t25] {};
        \node[redCan] (t45) [right=of t35] {};
        \node[blueStone] (t55) [right=of t45] {};
        \node[empty] (t65) [right=of t55] {};
        \node[redStone] (t75) [right=of t65] {};

        \node[blueStone] (t06) at (0, -4.2) {};
        \node[empty] (t16) [right=of t06] {};
        \node[empty] (t26) [right=of t16] {};
        \node[blueStone] (t36) [right=of t26] {};
        \node[empty] (t46) [right=of t36] {};
        \node[redCan] (t56) [right=of t46] {};
        \node[empty] (t66) [right=of t56] {};
        \node[nix] (t76) [right=of t66] {};

        \node[nix] (t07) at (-.5, -4.9) {};
        \node[redStone] (t17) [right=of t07] {};
        \node[blueStone] (t27) [right=of t17] {};
        \node[redStone] (t37) [right=of t27] {};
        \node[nix] (t47) [right=of t37] {};
        \node[nix] (t57) [right=of t47] {};
        \node[nix] (t67) [right=of t57] {};
        \node[blueStone] (t77) [right=of t67] {};
    \end{tikzpicture}\end{center}
    \caption{REDPLAYS gadget.  Blue cannot play here at all.  The vertices indicated with red hatching form a maximum independent set, so Red can move 8 times.}
    \label{fig:redPlays}
\end{figure}

In addition, we can color in all those vertices red and use the resulting unplayable REDPLAYS gadget as a filler to cover blank tiles between other gadgets in the circuitry.  Moreover, we can use this as a square border to our (enormous, but polynomial) board.  The only changes that need to be made are that when placed along the top of the entire board, the middle X space along the top of the gadget needs to be colored blue; and when placed along the bottom, the middle X space along the bottom of the gadget needs to be colored blue.

\begin{theorem}[Main]
    \rsCol\ is \cclass{PSPACE}-complete on triangular grid graphs on rectangles where sides are multiples of 8.
\end{theorem}

\begin{proof}
    To prove \cclass{PSPACE}-hardness, we reduce from \ruleset{Bounded 2-Player Constraint Logic} (\rsBCL) \cite{DBLP:books/daglib/0023750}.  The gadgets described above create the necessary pieces to complete the reduction from required \rsBCL\ vertices.  The gadgets each occupy $8 \times 8$ grids, so the resulting triangular grid-graph will have sides that are multiples of 8.

    \rsCol\ is in \cclass{PSPACE} because the height of the game tree is bounded by the number of vertices in the graph and the number of options from each position (fanout) is bounded by the same number.  Thus, the game can be solved in \cclass{PSPACE}.
\end{proof}

\subsection{Reachability}

The games generated by our reduction are not necessarily reachable from an empty grid, because they may not have the same number of red and blue-colored vertices at the beginning: Red has more plays made in the gadgets.  In order to solve that problem, we add another gadget that can be used to make up that difference.  These ``makeup'' tiles, as shown in Figure \ref{fig:makeup}, are unplayable tiles that have three more blue-colored vertices than red.  

We can use these to cover the difference between plays.  If that difference is not divisible by 3, we can include additional completely-filled-in REDPLAYS gadgets to reach the correct difference.  Now these starting boards have the same number of red and blue-colored vertices and is thus reachable in a real game from the empty grid.


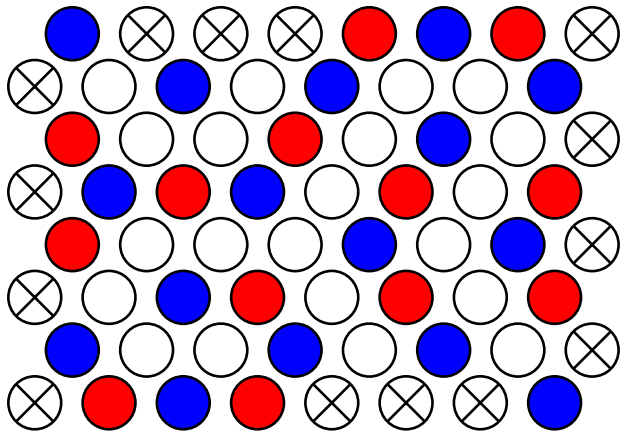
\begin{figure}[h!]
    \begin{center}
    \begin{tikzpicture}[node distance = .25cm, minimum size = .5cm, inner sep = .07cm, ultra thick]
        \node[blueStone] (t00) {};
        \node[nix] (t10) [right=of t00] {};
        \node[nix] (t20) [right=of t10] {};
        \node[nix] (t30) [right=of t20] {};
        \node[redStone] (t40) [right=of t30] {};
        \node[blueStone] (t50) [right=of t40] {};
        \node[redStone] (t60) [right=of t50] {};
        \node[nix] (t70) [right=of t60] {};

        \node[nix] (t01) at (-.5, -.7) {};
        \node[empty] (t11) [right=of t01] {};
        \node[blueStone] (t21) [right=of t11] {};
        \node[empty] (t31) [right=of t21] {};
        \node[blueStone] (t41) [right=of t31] {};
        \node[empty] (t51) [right=of t41] {};
        \node[empty] (t61) [right=of t51] {};
        \node[blueStone] (t71) [right=of t61] {};
        
        \node[redStone] (t02) at (0, -1.4) {};
        \node[empty] (t12) [right=of t02] {};
        \node[empty] (t22) [right=of t12] {};
        \node[redStone] (t32) [right=of t22] {};
        \node[empty] (t42) [right=of t32] {};
        \node[blueStone] (t52) [right=of t42] {};
        \node[empty] (t62) [right=of t52] {};
        \node[nix] (t72) [right=of t62] {};

        \node[nix] (t03) at (-.5, -2.1) {};
        \node[blueStone] (t13) [right=of t03] {};
        \node[redStone] (t23) [right=of t13] {};
        \node[blueStone] (t33) [right=of t23] {};
        \node[empty] (t43) [right=of t33] {};
        \node[redStone] (t53) [right=of t43] {};
        \node[empty] (t63) [right=of t53] {};
        \node[redStone] (t73) [right=of t63] {};
        
        \node[redStone] (t04) at (0, -2.8) {};
        \node[empty] (t14) [right=of t04] {};
        \node[empty] (t24) [right=of t14] {};
        \node[empty] (t34) [right=of t24] {};
        \node[blueStone] (t44) [right=of t34] {};
        \node[empty] (t54) [right=of t44] {};
        \node[blueStone] (t64) [right=of t54] {};
        \node[nix] (t74) [right=of t64] {};

        \node[nix] (t05) at (-.5, -3.5) {};
        \node[empty] (t15) [right=of t05] {};
        \node[blueStone] (t25) [right=of t15] {};
        \node[redStone] (t35) [right=of t25] {};
        \node[empty] (t45) [right=of t35] {};
        \node[redStone] (t55) [right=of t45] {};
        \node[empty] (t65) [right=of t55] {};
        \node[redStone] (t75) [right=of t65] {};

        \node[blueStone] (t06) at (0, -4.2) {};
        \node[empty] (t16) [right=of t06] {};
        \node[empty] (t26) [right=of t16] {};
        \node[blueStone] (t36) [right=of t26] {};
        \node[empty] (t46) [right=of t36] {};
        \node[blueStone] (t56) [right=of t46] {};
        \node[empty] (t66) [right=of t56] {};
        \node[nix] (t76) [right=of t66] {};

        \node[nix] (t07) at (-.5, -4.9) {};
        \node[redStone] (t17) [right=of t07] {};
        \node[blueStone] (t27) [right=of t17] {};
        \node[redStone] (t37) [right=of t27] {};
        \node[nix] (t47) [right=of t37] {};
        \node[nix] (t57) [right=of t47] {};
        \node[nix] (t67) [right=of t57] {};
        \node[blueStone] (t77) [right=of t67] {};
    \end{tikzpicture}\end{center}
    \caption{A makeup tile, which has three more plays made by Blue than Red.  (No one can play in these tiles.)}
    \label{fig:makeup}
\end{figure}

Additionally, instead of using REDPLAYS gadgets to fill in the border and extra space, we can use EVEN gadgets, shown in Figure \ref{fig:even} as unplayable tiles that have 14 blue and 14 red vertices.

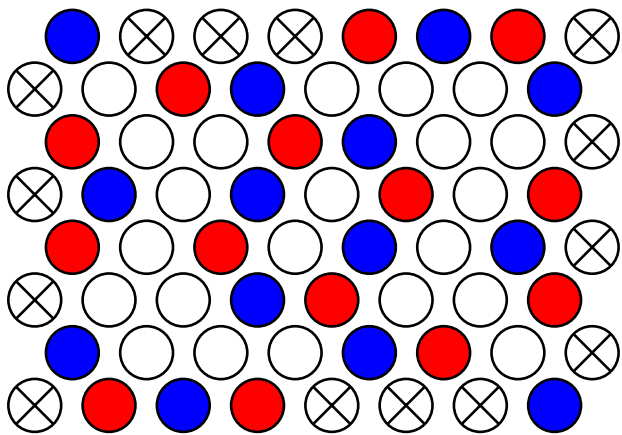
\begin{figure}[h!]
    \begin{center}
    \begin{tikzpicture}[node distance = .25cm, minimum size = .5cm, inner sep = .07cm, ultra thick]
        \node[blueStone] (t00) {};
        \node[nix] (t10) [right=of t00] {};
        \node[nix] (t20) [right=of t10] {};
        \node[nix] (t30) [right=of t20] {};
        \node[redStone] (t40) [right=of t30] {};
        \node[blueStone] (t50) [right=of t40] {};
        \node[redStone] (t60) [right=of t50] {};
        \node[nix] (t70) [right=of t60] {};

        \node[nix] (t01) at (-.5, -.7) {};
        \node[empty] (t11) [right=of t01] {};
        \node[redStone] (t21) [right=of t11] {};
        \node[blueStone] (t31) [right=of t21] {};
        \node[empty] (t41) [right=of t31] {};
        \node[empty] (t51) [right=of t41] {};
        \node[empty] (t61) [right=of t51] {};
        \node[blueStone] (t71) [right=of t61] {};
        
        \node[redStone] (t02) at (0, -1.4) {};
        \node[empty] (t12) [right=of t02] {};
        \node[empty] (t22) [right=of t12] {};
        \node[redStone] (t32) [right=of t22] {};
        \node[blueStone] (t42) [right=of t32] {};
        \node[empty] (t52) [right=of t42] {};
        \node[empty] (t62) [right=of t52] {};
        \node[nix] (t72) [right=of t62] {};

        \node[nix] (t03) at (-.5, -2.1) {};
        \node[blueStone] (t13) [right=of t03] {};
        \node[empty] (t23) [right=of t13] {};
        \node[blueStone] (t33) [right=of t23] {};
        \node[empty] (t43) [right=of t33] {};
        \node[redStone] (t53) [right=of t43] {};
        \node[empty] (t63) [right=of t53] {};
        \node[redStone] (t73) [right=of t63] {};
        
        \node[redStone] (t04) at (0, -2.8) {};
        \node[empty] (t14) [right=of t04] {};
        \node[redStone] (t24) [right=of t14] {};
        \node[empty] (t34) [right=of t24] {};
        \node[blueStone] (t44) [right=of t34] {};
        \node[empty] (t54) [right=of t44] {};
        \node[blueStone] (t64) [right=of t54] {};
        \node[nix] (t74) [right=of t64] {};

        \node[nix] (t05) at (-.5, -3.5) {};
        \node[empty] (t15) [right=of t05] {};
        \node[empty] (t25) [right=of t15] {};
        \node[blueStone] (t35) [right=of t25] {};
        \node[redStone] (t45) [right=of t35] {};
        \node[empty] (t55) [right=of t45] {};
        \node[empty] (t65) [right=of t55] {};
        \node[redStone] (t75) [right=of t65] {};

        \node[blueStone] (t06) at (0, -4.2) {};
        \node[empty] (t16) [right=of t06] {};
        \node[empty] (t26) [right=of t16] {};
        \node[empty] (t36) [right=of t26] {};
        \node[blueStone] (t46) [right=of t36] {};
        \node[redStone] (t56) [right=of t46] {};
        \node[empty] (t66) [right=of t56] {};
        \node[nix] (t76) [right=of t66] {};

        \node[nix] (t07) at (-.5, -4.9) {};
        \node[redStone] (t17) [right=of t07] {};
        \node[blueStone] (t27) [right=of t17] {};
        \node[redStone] (t37) [right=of t27] {};
        \node[nix] (t47) [right=of t37] {};
        \node[nix] (t57) [right=of t47] {};
        \node[nix] (t67) [right=of t57] {};
        \node[blueStone] (t77) [right=of t67] {};
    \end{tikzpicture}\end{center}
    \caption{An EVEN tile, unplayable with 14 blue and 14 red vertices.  These can be used to fill in the grid without altering the balance in the number of moves each player has made.  Additionally, no modifications must be made when placing along the border.}
    \label{fig:even}
\end{figure}

The process to perform the reduction now has a few more steps:

\begin{enumerate}
    \item Build the gadget tiles from \rsBCL and connect them for the main construction.
    \item Add REDPLAYS tiles so that Blue has exactly one more play than Red if they are able to play twice on the GOAL gadget.
    \item Add MAKEUP tiles and filled-in REDPLAYS tiles until the number of blue and red vertices is equivalent.
    \item Fill in any remaining holes in the grid with EVEN tiles and use them to create a border around the entire grid.
\end{enumerate}

\section*{Acknowledgements}

Thank you to our families for understanding that we are very excited about studying a basic graph game.

\bibliographystyle{plainurl}

\end{document}